\documentclass[a4paper]{article}

\usepackage{pstricks}
\usepackage{pst-node}
\usepackage[active]{srcltx}
\usepackage[utf8]{inputenc} 
\usepackage[T1]{fontenc}
\usepackage[english]{babel}
\usepackage{comment}
\usepackage{epsfig}
\usepackage{amssymb,amsmath,amsthm}
\usepackage{xspace}
\usepackage{enumerate}
\usepackage{longtable}

\theoremstyle{plain}
\newtheorem{theorem}{Theorem}
\newtheorem{lemma}[theorem]{Lemma}
\newtheorem{fact}[theorem]{Fact}
\newtheorem{proposition}[theorem]{Proposition}

\theoremstyle{definition}

\newtheorem{conjecture}{Conjecture}

\theoremstyle{remark}

\theoremstyle{remark}

\date{February 24, 2011}

\begin{document}
\title{\bf The Domination Number of Grids\thanks{This work has been
    partially supported by the ANR Project GRATOS ANR-JCJC-00041-01.}}
\author{Daniel Gon\c calves\thanks{Universit\'e Montpellier 2, CNRS, LIRMM
    161 rue Ada,
    34095 Montpellier Cedex, France. 
    Email: \texttt{\{daniel.goncalves | alexandre.pinlou | stephan.thomasse\}@lirmm.fr}}\addtocounter{footnote}{-1}
  \and Alexandre Pinlou\footnotemark\ \ \thanks{Département Mathématiques
    et Informatique Appliqués, Université Paul-Valéry, Montpellier 3,
    Route de Mende, 34199 Montpellier Cedex 5, France.}\and Micha\"el Rao\thanks{CNRS, Laboratoire J.-V. Poncelet, Moscow, Russia. 
    Universit\'e Bordeaux - CNRS,
    LABRI, 351, cours de la Lib\'eration 33405 Talence, France. 
    Email: \texttt{michael.rao@labri.fr}}\addtocounter{footnote}{-3} \and St\'ephan Thomass\'e\footnotemark } 

\maketitle

\def\si{\ \mbox{\rm{if}} \ }
\def\et{\ \mbox{\rm{and}} \ }
\def\ou{\ \mbox{\rm{or}} \ }
\def\form#1#2{\left \lfloor \frac{(#1+2)(#2+2)}{5} \right \rfloor -4}
\def\frceil#1#2{\lceil\frac{#1}{#2}\rceil}
\def\frfloor#1#2{\lfloor\frac{#1}{#2}\rfloor}
\def\ceil#1{\lceil#1\rceil}
\def\fG#1#2{G_{#1,#2}}
\def\gG#1#2{\gamma(G_{#1,#2})}
\def\Gnm{G_{n,m}}

\def\fV#1#2{V_{#1,#2}}
\def\Vnm{V_{n,m}}

\def\los#1{\ell(#1)}
\def\lnm{\ell_{n,m}}
\def\lif{l_{i,f}}

\def\matS{C}
\def\matAS{M}
\def\matMin{M'}

\def\matL{L}
\def\matT{T}

\def\bnm{b_{n,m}}

\def\fB#1#2{B_{#1,#2}}
\def\Bnm{B_{n,m}}

\def\fBp#1#2{I(B_{#1,#2})}
\def\Bpnm{I(B_{n,m})}

\def\win{w^{in}}
\def\wout{w^{out}}

\def\mF{\mathcal{F}}

\def\fP#1{P_{#1}}
\def\Pp{P_p}
\def\Ppp{I(P_p)}

\def\fQ#1{Q_{#1}}
\def\Qp{Q_p}
\def\fR#1{R_{#1}}
\def\Rp{R_p}
\def\fS#1{O_{#1}}
\def\Sp{O_p}

\def\fC#1{C_{#1}}
\def\Op{O_p}
\def\Ip{I_p}
\def\Cq{C_q}

\def\fnm{f_{n,m}}
\def\fmn{f_{m,n}}

\def\mW{\mathcal{W}}
\def\mD{\mathcal{D}}

\def\Dwwp{\mD^{w,w'}_p}
\newcommand{\Dww}[1]{\mD^{w,w'}_{#1}}

\begin{abstract} In this paper, we conclude the calculation of the
  domination number of all $n\times m$ grid graphs. Indeed, we prove
  Chang's conjecture saying that for every $16\le n\le m$,
  $\gamma(\Gnm)=\form{n}{m}$. 
\end{abstract}

\section{Introduction}

A {\it dominating set} in a graph $G$ is a subset 
of vertices $S$ such that every vertex in $V(G)\setminus S$
is a neighbour of some vertex of $S$. The {\it domination number} 
of $G$ is the minimum size of a dominating set of $G$. We 
denote it by $\gamma(G)$. This paper is devoted to the calculation 
of the domination number of complete grids.

The notation $[i]$ denotes the set $\{1,2,\ldots, i\}$. If $w$ is a
word on the alphabet $A$, $w[i]$ is the $i$-th letter of $w$, and 
for every $a$ in  $A$, $\vert w \vert_a$ denotes the number of occurrences of 
$a$ in $w$ (i.e. $\vert \{ i\in \{1,\ldots,\vert w\vert\} 
: w[i]=a\}\vert$). 
For a vertex $v$, $N[v]$ denotes the closed neighbourhood of $v$ 
(i.e. the set of
neighbours of $v$ and $v$ itself). For a subset of
vertices $S$ of a vertex set $V$ of a graph, we denote by $N[S] =
\bigcup_{v\in S}N[v]$.  Note that $D$ is a dominating set of $G$ if
and only if $N[D] = V(G)$. Let $\Gnm$ be the $n\times m$ complete
grid, i.e. the vertex set of $\Gnm$ is $\Vnm:=[n]\times [m]$, and two
vertices $(i,j)$ and $(k,l)$ are adjacent if
$|k-i|+|l-j|=1$. The couple $(1,1)$ denotes the bottom-left vertex of
the grid and the couple $(i,j)$ denotes the vertex of the $i$-th column
and the $j$-th row. We will always assume in this paper that $n\leq
m$. 
Let us illustrate our purpose by an example of a dominating set of the
complete grid $\fG{24}{24}$ on Figure~\ref{fig:24x24}.

\begin{figure}
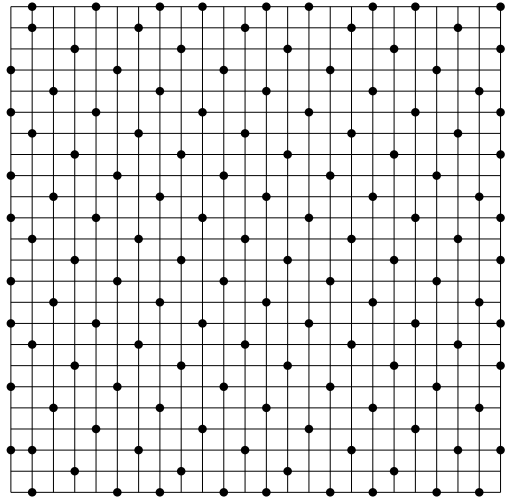

  \centering
  \psset{unit=0.28cm}
  \include{d24x24}
  \caption{Example of a set of size 131 dominating the grid $\fG{24}{24}$\label{fig:24x24}}
\end{figure}

The first results on the domination number of grids were obtained
about 30 years ago with the exact values of $\gG2n$, $\gG3n$, and $\gG4n$
found by Jacobson and Kinch \cite{JK} in 1983. In 1993,  
Chang and Clark~\cite{CC} found those of $\gG5n$
and $\gG6n$. These results were obtained analytically.
Chang~\cite{Chang} devoted his PhD thesis to study the domination
number of grids; he conjectured that this invariant
behaves well provided that $n$ is large enough.
Specifically, Chang conjectured the following:

\begin{conjecture}[\cite{Chang}]\label{conj}
  For every $16\le n \le m$, $$\gamma(\Gnm) = \form{n}{m}.$$
\end{conjecture}

Observe that for instance, this formula would give 131 for the domination number 
of the grid in Figure~\ref{fig:24x24}. To motivate his bound, Chang proposed some constructions of
dominating sets achieving the upper bound:

\begin{lemma}[\cite{Chang}]\label{lem:cgang}
  For every $8\le n\le m$, $$\gamma(\Gnm) \le \form{n}{m}$$
\end{lemma}

Later, some algorithms based on dynamic programming were designed to
compute a lower bound of $\gG{n}{m}$. There were numerous
intermediate results found for $\gG{n}{m}$ for small values of
$n$ and $m$ (see \cite{CCH,HHH,SP} for details). In 1995, Hare,
Hedetniemi and Hare \cite{HHH} gave a polynomial time algorithm to
compute $\gamma(\Gnm)$ when $n$ is fixed. Nevertheless, this
algorithm is not usable in practice when $n$ hangs over $20$. Fisher
\cite{Fisher} developed the idea of searching for periodicity in the
dynamic programming algorithms and using this technique, he found the
exact values of $\gG{n}{m}$ for all $n\leq 21$. We recall these values
for the sake of completeness.

\begin{theorem}[\cite{Fisher}]
  For all $n \le m$ and $n \le 21$, we have:
  $$\gG{n}{m} = \left\{
    \begin{array}{ll}
      \frceil{m}{3} & \si n=1 \\*[0.1cm]
      \frceil{m+1}{2} & \si n=2 \\*[0.1cm]
      \frceil{3m+1}{4} & \si n=3 \\*[0.1cm]
      m+1 & \si n=4  \et m=5,6,9\\*[0.1cm]
      m & \si n=4  \et m\neq 5,6,9\\*[0.1cm]
      \frceil{6m+4}{5} - 1 & \si n=5 \et m=7 \\*[0.1cm]
      \frceil{6m+4}{5} & \si n=5 \et m\neq 7 \\*[0.1cm]
      \frceil{10m+4}{7} & \si n=6 \\*[0.1cm]
      \frceil{5m+1}{3} & \si n=7 \\*[0.1cm]
      \frceil{15m+7}{8} & \si n=8 \\*[0.1cm]
      \frceil{23m+10}{11} & \si n=9 \\*[0.1cm]
      \frceil{30m+15}{13} - 1 & \si n=10 \et m\equiv_{13}10 \ou m=13,16
      \\*[0.1cm]
      \frceil{30m+15}{13} & \si n=10 \et m\not\equiv_{13}10 \et m\neq13,16 \\*[0.1cm]
      \frceil{38m+22}{15} - 1 & \si n=11 \et m=11,18,20,22,33 \\*[0.1cm]
      \frceil{38m+22}{15} & \si n=11 \et m\neq11,18,20,22,33 \\*[0.1cm]
      \frceil{80m+38}{29} & \si n=12 \\*[0.1cm]
      \frceil{98m+54}{33} - 1 & \si n=13 \et m\equiv_{33}13,16,18,19 \\*[0.1cm]
      \frceil{98m+54}{33} & \si n=13 \et m\not\equiv_{33}13,16,18,19 \\*[0.1cm]
      \frceil{35m+20}{11} - 1 & \si n=14 \et m\equiv_{22}7\\*[0.1cm]
      \frceil{35m+20}{11} & \si n=14 \et m\not\equiv_{22}7\\*[0.1cm]
      \frceil{44m+28}{13} - 1 & \si n=15 \et  m\equiv_{26}5\\*[0.1cm]
      \frceil{44m+28}{13} & \si n=15 \et  m\not\equiv_{26}5\\*[0.1cm]
      \frfloor{(n+2)(m+2)}{5}-4 & \si n\ge16 \\*[0.1cm]
    \end{array}\right.$$
  
\end{theorem}

Note that these values are obtained by specific formulas for every $1\leq n\leq 15$ and by
the formula of Conjecture~\ref{conj} for every $16\le n \le 21$. This
proves Chang's conjecture for all values $16\le n\le 21$.

In 2004, Conjecture~\ref{conj} has been confirmed up to an 
additive constant:

\begin{theorem}[Guichard \cite{Guichard}] For every $16\le n \le m$,
  $$\gamma(\Gnm) \geq \left \lfloor \frac{(n+2)(m+2)}{5} \right
  \rfloor -9.$$
\end{theorem}

In this paper, we prove Chang's conjecture, hence finishing the
computation of $\gG{n}{m}$.  We adapt Guichard's ideas to improve the
additive constant from $-9$ to $-4$ when $24 \le n \le m$.  Cases
$n=22$ and $n=23$ can be proved in a couple of hours using Fisher's
method (described in~\cite{Fisher}) on a modern computer.  They can be
also proved by a slight improvement of the technique presented in the
next section.

\section{Values of $\gamma(\Gnm)$ when $24 \le n \le m$}

Our method follows the idea of Guichard \cite{Guichard}. A slight
improvement is enough to give the exact bound.  

A vertex of the grid $\Gnm$ dominates at most 5 vertices (its four
neighbours and itself). It is then clear that $\gG{n}{m} \ge
\frac{n\times m}{5}$. The previous inequality would become an equality
if there would be a dominating set $D$ such that every vertex of
$\Gnm$ is dominated only once, and all vertices of $D$ are of degree 4
(i.e.  dominates exactly 5 vertices); in this case, we would have
$5\times|D| - n\times m= 0$. This is clearly impossible (e.g. to
dominate the corners of the grid, we need vertices of degree at most
3). Therefore, our goal is to find a dominating set $D$ of $\Gnm$ such
that the difference $5\times |D| - n\times m$ is the smallest.

Let $S$ be a subset of $V(\Gnm)$. The \emph{loss} of $S$ is $\los{S}=5
\times \vert S \vert - \vert N[S] \vert$.

\begin{proposition} \label{props}
  The following properties of the loss function are straightforward:
  \begin{enumerate}[(i)]
  \item For every $S$, $\los{S} \ge 0$ (positivity), \label{props:1}
  \item If $S_1\cap S_2=\emptyset$, then $\los{S_1\cup S_2} =
    \los{S_1} + \los{S_2} + \vert N[S_1] \cap N[S_2] \vert$, \label{props:2}
  \item If $S'\subseteq S$, then $\los{S'}\le \los{S}$ (increasing
    function), \label{props:3}
  \item If $S_1\cap S_2=\emptyset$, then $\los{S_1\cup S_2} \ge
    \los{S_1} + \los{S_2}$ (super-additivity). \label{props:4}
  \end{enumerate}
\end{proposition}

Let us denote by $\lnm$ the minimum of $\los{D}$ when $D$ is 
a dominating set of $\Gnm$. 

\begin{lemma}\label{lemloss}
  $\gamma(\Gnm) = \left \lceil \frac{n\times m + \lnm}{5} \right
  \rceil$
\end{lemma}

\begin{proof}
  If $D$ is a dominating set of $\Gnm$, then $\los{D} = 5 \times \vert
  D \vert - \vert N[D]\vert = 5 \times \vert
  D \vert - n\times m$. Hence, by minimality of $\lnm$ and
  $\gamma(\Gnm)$, we have $\lnm= 5\times \gamma(\Gnm) - n\times m$.
\end{proof}

Our aim is to get a lower bound for $\lnm$. As the reader can observe
in Figure~\ref{fig:24x24}, the loss is concentrated on the border of
the grid. We now analyse more carefully the loss generated by the
border of thickness $10$. 

\begin{figure}
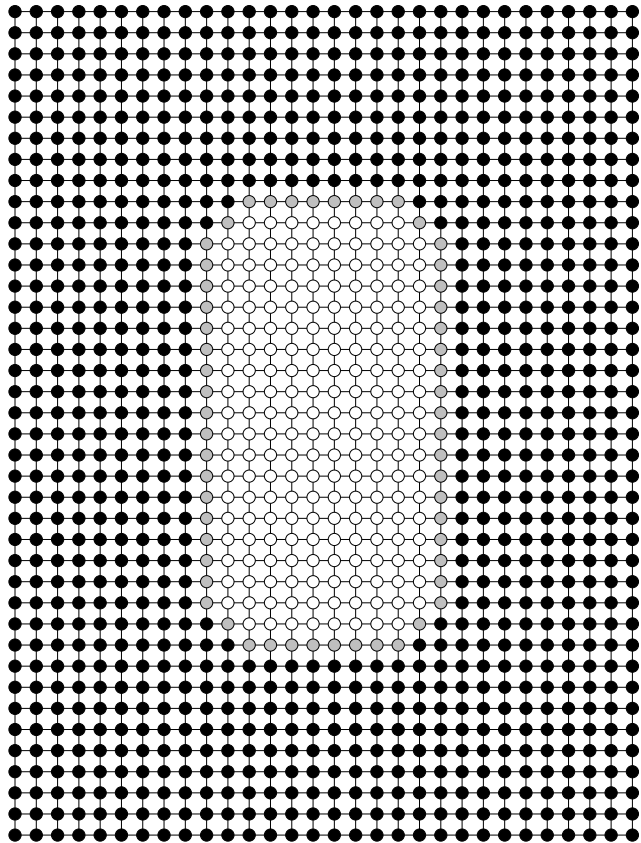

  \centering
  \psset{unit=0.28cm}
  \begin{minipage}{0.9\linewidth}
    \centering \include{B}
  \end{minipage}
  \caption{The graph $\fG{30}{40}$. The set $\fBp{30}{40}$ is the set
    of vertices filled in black. The set $\fB{30}{40}$ is the set of
    vertices filled in black or in gray.\label{fig:bnm}}
\end{figure}

We define the border $\Bnm \subseteq \Vnm$ of $\Gnm$ as the set of
vertices $(i,j)$ where $i\le 10$, or $j\le 10$, or $i>n-10$, or
$j>m-10$ to which we add the four vertices
$(11,11),(11,m-10),(n-10,11),(n-10,m-10)$.  Given a subset $S\subseteq
V$, let $I(S)$ be the \emph{internal vertices} of $S$, i.e. $I(S)=
\{v\in S : N[v]\subseteq S\}$.  These sets are illustrated in
Figure~\ref{fig:bnm}.  We will compute $\bnm = \min_D \los{D}$, where
$D$ is a subset of $\Bnm$ and dominates $\Bpnm$, i.e. $D\subseteq
\Bnm$ and $\Bpnm \subseteq N[D]$. Observe that this lower bound $\bnm$
is a lower bound of $\lnm$. Indeed, for every dominating set $D$ of
$\Gnm$, the set $D':=D\cap \Bnm$ dominates $\Bpnm$, hence $\bnm\leq
\los{D'}\leq \los{D}$. In the remainder, we will compute $\bnm$ and we
will show that $\bnm = \lnm$.

\bigskip

In the following, we split the border $\Bnm$ in four parts, $O_{m-12},
P_{n-12}, Q_{m-12}, R_{n-12}$, which are defined just below.

\begin{figure}
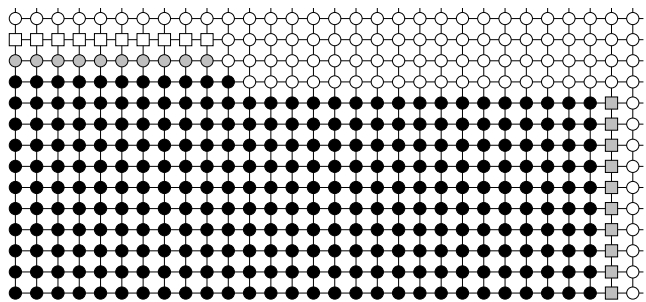

  \centering
  \psset{unit=0.28cm}
  \include{Q}
  \caption{The set $\fP{19}$ (black and gray), the set of input vertices
    (gray circles) and the set of output vertices (gray squares).\label{fig:ppcpa}}
\end{figure}

For $p\ge 12$, let $\Pp\subset \Bnm$ defined as follows : $\Pp = ([10]
\times \{12\}) \cup ([11] \times \{11\} ) \cup ([p] \times [10])$. We
define the \emph{input vertices} of $\Pp$ as $[10]\times\{12\}$ and
the \emph{output vertices} of $\Pp$ as $\{p\} \times [10]$. The set
$\Pp$, illustrated for $p=19$ in Figure~\ref{fig:ppcpa}, corresponds
to the set of black and gray vertices. The input vertices are the gray
circles, and the output vertices are the gray squares. Recall that in
our drawing conventions, the vertex $(1,1)$ is the bottom-left vertex
and hence the vertex $(i,j)$ is in the $i^{th}$ column from the left
and in the $j^{th}$ row from the bottom.

\begin{figure}
  \centering
  \psset{unit=0.28cm}
  \begin{minipage}{0.9\linewidth}
    \centering \begin{pspicture}(-0.5,-0.5)(30.5,40.5)
\psline[linecolor=black,linewidth=0.025](0,0)(30,0)(30,40)(0,40)(0,0)
\psline[linecolor=black,linewidth=0.025,linestyle=dashed](12.50,10.50)(11.50,10.50)(11.50,11.50)(10.50,11.50)(10.50,27.50)
\psline[linecolor=black,linewidth=0.025,linestyle=dashed](0.00,12.5)(10.5,12.5)
\rput(10,5){\hbox{\Large $P_{n-12}$}}
\rput(20,35){\hbox{\Large $R_{n-12}$}}
\rput(6,25){\hbox{\Large $Q_{m-12}$}}
\rput(25,15){\hbox{\Large $O_{m-12}$}}
{\rput*[0,0]{90}(16,15){\hbox{\Large $V(G_{n,m}) \setminus B_{n,m}$}}}
\psline[linecolor=black,linewidth=0.025,linestyle=dashed](19.50,12.50)(19.50,11.50)(18.50,11.50)(18.50,10.50)(12.50,10.50)
\psline[linecolor=black,linewidth=0.025,linestyle=dashed](17.50,0.0)(17.5,10.5)
\psline[linecolor=black,linewidth=0.025,linestyle=dashed](17.50,29.50)(18.50,29.50)(18.50,28.50)(19.50,28.50)(19.50,12.50)
\psline[linecolor=black,linewidth=0.025,linestyle=dashed](19.50,27.5)(30,27.5)
\psline[linecolor=black,linewidth=0.025,linestyle=dashed](10.50,27.50)(10.50,28.50)(11.50,28.50)(11.50,29.50)(17.50,29.50)
\psline[linecolor=black,linewidth=0.025,linestyle=dashed](12.5,29.5)(12.5,40)
\end{pspicture}
  \end{minipage}
  \caption{The sets $O_{m-12}$, $P_{n-12}$, $Q_{m-12}$ and $R_{n-12}$.
    \label{fig:pqrs}}
\end{figure}
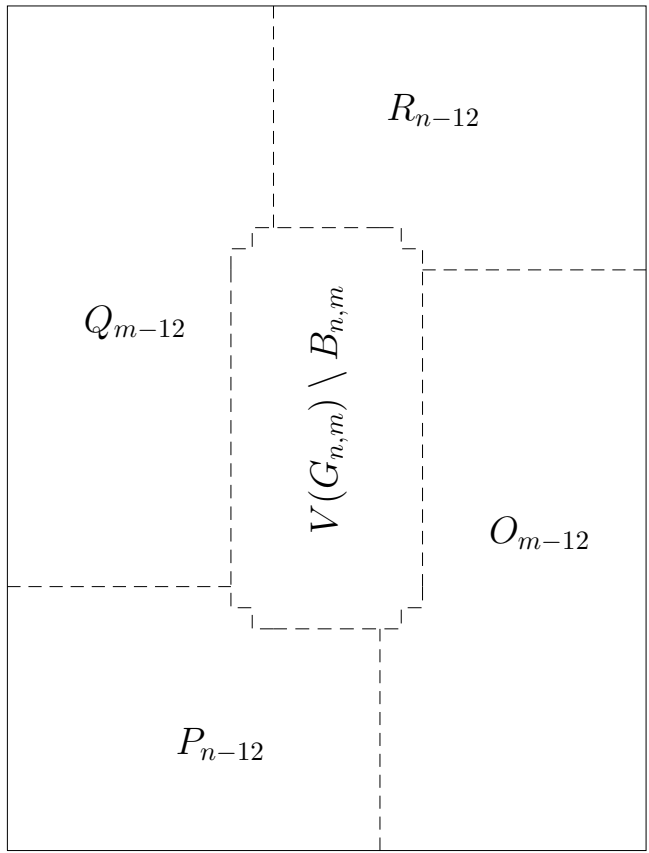

For $n,m\in \mathbb{N}^*$, let $\fnm: [n]\times[m] \to [m]\times[n]$ be
the bijection such that $\fnm(i,j) = (j,n-i+1)$. It is clear that the
set $\Bnm$ is the disjoint union of the following four sets depicted
in Figure~\ref{fig:pqrs}: $\fP{n-12}$, $\fQ{m-12}=\fnm(\fP{m-12})$,
$\fR{n-12} = \fmn \circ \fnm (\fP{n-12})$ and $\fS{m-12}= \fnm^{-1}(\fP{m-12})$. Similarly to $\fP{n-12}$, the sets
$O_{m-12}$, $Q_{n-12}$ and $R_{m-12}$ have input and output
vertices. For instance, the output vertices of $Q_{m-12}$ correspond
in Figure~\ref{fig:ppcpa} to the white squares. Every set playing a
symmetric role, we now focus on $P_{n-12}$.

\bigskip

Given a subset $S$ of $\Vnm$, let the labelling $\phi_S : \Vnm \to
\{0,1,2\}$ be such that 

$$\phi_S(i,j) = \left\{ 
  \begin{array}{l}
    0 \quad \mbox{if $(i,j)\in S$} \\
    1 \quad \mbox{if $(i,j)\in N[S]\setminus S$} \\
    2 \quad \mbox{otherwise}
  \end{array}\right.$$

Note that $\phi_S$ is such that any two
adjacent vertices in $\Gnm$ cannot be labelled~$0$ and~$2$.

Given $p\ge 12$ and a set $S\subseteq \Pp$, the \emph{input word} 
(resp. \emph{output word}) of $S$ for $P_p$, denoted by $\win(S)$
(resp. $\wout_p(S)$), is the ten letters word on the alphabet
$\{0,1,2\}$ obtained by reading $\phi_{S}$ on the input vertices 
(resp. output vertices) of $P_p$.
More precisely, its $i^{th}$ letter is $\phi_{S}(i,12)$ (resp. $\phi_{S}(p,i)$). 
Similarly, $O_p$, $Q_p$ and $R_p$ have also input and output words. 
For example, the output word of $S\subseteq O_p$ for $O_p$ is $\wout_p(\fnm(S))$.

According to the definition of $\phi$, the input and output words
belong to the set $\mW$ of ten letters words on $\{0,1,2\}$ which
avoid $02$ and $20$. The number of $k$-digits trinary numbers without
$02$ or $20$ is given by the following formula~\cite{Fisher}:
\begin{equation}
  \frac{(1+\sqrt{2})^{k+1}+(1-\sqrt{2})^{k+1}}{2}\label{eqn:trinary_words}
\end{equation}
The size of $\mW$ is therefore $|\mW| = 8119$.

Given two words $w,w' \in \mW$, we define
$\Dwwp$ as the family of subsets $D$ of $\fP{p}$ such that:
\begin{itemize}
\item $D$ dominates $\Ppp$,
\item $w$ is the input word $\win(D)$,
\item $w'$ is the output word $\wout_p(D)$.
\end{itemize}

A relevant information for our calculation will be to know, for two
given words $w,w'\in \mW$, the minimum loss over all losses $\los{D}$
where $D\in \Dwwp$.  For this purpose, we introduce the $8119\times
8119$ square matrix $C_p$. For $w,w'\in \mW$, let
$\matS_p[w,w']=\min_{D\in\Dwwp} \los{D}$ where the minimum of the
empty set is $+\infty$.

Let $w,w'\in \mW$ be two given words. Due to the symmetry of $P_{12}$
with respect to the first diagonal (bottom-left to top-right) of the
grid, if a vertex set $D$ belongs to $\Dww{12}$, then $D' = \{(j,i) |
(i,j)\in D\}$ belongs to $\mD^{w',w}_{12}$. Moreover, it is clear
that, always due to the symmetry, $\los{D}=\los{D'}$. Therefore, we
have $\matS_{12}[w,w'] = \matS_{12}[w',w]$ and thus $\matS_{12}$ is a
symmetric matrix.  Despite the size of $\matS_{12}$ and the size of
$P_{12}$ (141 vertices), it is possible to compute $\matS_{12}$ in
less than one hour by computer. Indeed, we choose a sequence of subsets
$X_0=\emptyset, X_1, \ldots, X_{141}=P_{12}$ such that for every $i\in
\{1,\ldots, 141\}$, $X_i\subseteq X_{i+1}$ and $X_{i+1} \setminus X_i
= \{ x_{i+1} \}$.  Moreover, we choose the sequence such that for
every $i$, $X_i \setminus I(X_i)$ is at most $21$. This can be done
for example by taking $x_{i+1}= \min \{ (x,y) : (x,y)\in P_{12}
\setminus X_i \}$, where the order is the lexical order. For $i\in
\{0,\ldots, 141\}$, we compute for every labeling $f \in \mF_i$, where
$\mF_i$ is the set of functions $(X_i\setminus I(X_i)) \to \{0,1,2\}$,
the minimal loss $\lif$ of a set $S\subseteq X_i$ which dominates
$I(X_i)$ and such that $\phi_S(v)=f(v)$ for every $v\in X_i \setminus
I(X_i)$. Note that not every labeling is possible since two adjacent
vertices cannot be labeled~$0$ and~$2$. The number of possible
labellings can be computed using formula~(\ref{eqn:trinary_words}),
and since $|X_i \setminus I(X_i)|$ can be covered by a path of at most
$23$ vertices, this gives, in the worst case, that this number is less
than $10^9$ and can be then processed
by a computer. We compute inductively the sequence $(l_{i,f})_{f\in
  \mF_i}$ from the sequence $(l_{i-1,f})_{f\in \mF_{i-1}}$ by dynamical
programming, and $\matS$ is easily deduced from $(l_{141,f})_{f\in
  \mF_{141}}$.

\bigskip

In the following, our aim is to glue $P_{n-12},
Q_{m-12},R_{n-12},$ and $O_{m-12}$ together. For two consecutive parts
of the border, say $P_{n-12}$ and $Q_{m-12}$, the output word of
$Q_{m-12}$ should be compatible with the input word of $P_{n-12}$.
Two words $w,w'$ of $\mW$ are \emph{compatible} if the sum of their
corresponding letters is at most 2, i.e. $w[i]+w'[i]\leq 2$ for all
$i\in [9]$. Note that $w[10] + w'[10]$ should be greater than 2 since
the corresponding vertices can be dominated by some vertices of $\Vnm
\setminus \Bnm$.

Given two words $w,w'\in \mW$, let
$\ell(w,w')= \vert \{ i\in [10] : w[i]\ne 2 \text{ and } w'[i]=0 \}
\vert + \vert \{ i\in [10] : w'[i]\ne 2 \text{ and } w[i]=0 \} \vert$.

\begin{lemma} \label{lem:compatible}Let $D$ be a dominating set of $\Gnm$ and let us denote
  $D_P = D\cap P_{n-12}$ and $D_Q = D\cap Q_{m-12}$. Then $\ell(D\cap
  (P_{n-12} \cup Q_{m-12}))= \ell(D_P) + \ell(D_Q) + \ell(w,w')$, where
  $w=\win(D_P)$ and $w'=\wout_q(\fnm^{-1}(D_Q))$. Moreover, $w$ and $w'$
  are compatible.
\end{lemma}

\begin{proof}
  By Proposition \ref{props}(\ref{props:2}), $\ell(D\cap (P_{n-12} \cup Q_{m-12}))=
  \ell(D_P) + \ell(D_Q) + \vert N[D_P] \cap N[D_Q]\vert$. It suffice then to note that
  $\ell(w,w')= \vert N[D_P] \cap N[D_Q]\vert$ to get $\ell(D\cap (P_{n-12} \cup Q_{m-12}))=
  \ell(D_P) + \ell(D_Q) + \ell(w,w')$.
  
  In what remains, we prove that $w$ and $w'$ are compatible.  If
  those two words were not compatible, there would exist an index $i\in[9]$
  such that $\wout_{m-12}(\fnm^{-1}(D_Q))[i] + \win(D_P)[i] > 2$.  Thus
  at least one of these two letters should be a 2, and the other one
  should not be 0.
  
  Suppose that $\wout_{m-12}(\fnm^{-1}(D_Q))[i] = 2$ and note that
  this means that the vertex $(i,13)$ is not dominated by a vertex in
  $D_Q$. Since $D$ is a dominating set of $\Gnm$, every output vertex of
  $Q_{m-12}$ except $(10,13)$ (and every input vertex of $P_{n-12}$
  except $(10,12)$) is dominated by a vertex of $D_Q$ or by a vertex of
  $D_P$. Thus $(i,13)$ should be dominated by its unique neighbour in
  $P_{n-12}$, $(i,12)$. This would imply that $(i,12)\in D$
  contradicting the fact that $\win(D_P)[i] \neq 0$.

  Similarly, if $\win(D_P)[i] =2$, the vertex $(i,12)$ 
  is not dominated by a vertex in $D_P$, thus $(i,12)$ must be dominated 
  by the vertex $(i,13)\in D$, contradicting the fact that
  $\wout_{m-12}(\fnm^{-1}(D_Q))[i] \neq 0$.
\end{proof}

Lemma~\ref{lem:compatible} is designed for the two consecutive parts
$P_{n-12}$ and $Q_{m-12}$ of the border of $\Gnm$. Its easy to see
that this extends to any pair of consecutive parts of the
border, i.e. $Q_{m-12}$ and $R_{n-12}$, $R_{n-12}$ and $O_{m-12}$,
$O_{m-12}$ and $P_{n-12}$.

\bigskip

We define the matrix $8119\times 8119$ square matrix $\matL$ which
contains, for every pair of words $w,w'\in \mW$, the value $\ell(w,w')$:   
$$\matL[w,w']=
\begin{cases} 
  +\infty \text{ if $w$ and $w'$ are not compatible,}\\
  \ell(w,w') \text{ otherwise.}
\end{cases}$$
Note that $\matL$ is symmetric since $\ell(w,w') = \ell(w',w)$.

Let $\otimes$ be the matrix multiplication in $(\min,+)$ algebra, 
i.e. $C = A \otimes B$ is the matrix where for all $i,j$, $C[i,j] =
\min_k A[i,k] + B[k,j]$. 

Let $\matAS_p=\matL \otimes \matS_p$ for $p\ge 12$. 

By construction, $\matAS_{n-12}[w,w']$ corresponds to the minimum
possible loss $\los{D\cap P_{n-12}}$ of a dominating set $D\subseteq
\Vnm$ that dominates $I(P_{n-12})$ and such that $w$ is the output
word of $Q_{m-12}$ and $w'$ is the output word of $P_{n-12}$.

\begin{lemma}
  \label{lem:bnm}
  For all $24 \le n \le m$, we have 
  $$b_{n,m}\geq \min_{w_1,w_2,w_3,w_4\in \mW} \matAS_{n-12}[w_1,w_2] +
  \matAS_{m-12}[w_2,w_3] + \matAS_{n-12}[w_3,w_4]+ \matAS_{m-12}[w_4,w_1].$$
\end{lemma}
\begin{proof} Consider a set $D \subseteq \Bnm$ which dominates
  $\Bpnm$ and achieving $\ell (D)=b_{n,m}$.  Let $D_P= D\cap P_{n-12}$,
  $D_Q= D\cap Q_{m-12}$, $D_R= D\cap R_{n-12}$ and $D_O= D\cap
  O_{m-12}$. Let $w_P$ ($w_Q$, $w_R$ and $w_O$, respectively) be the
  input word of $P_{n-12}$ ($Q_{m-12}$, $R_{n-12}$ and $O_{m-12}$), and
  $w'_P$ ($w'_Q$, $w'_R$ and $w'_O$) be the output word of $P_{n-12}$
  ($Q_{m-12}$, $R_{n-12}$ and $O_{m-12}$).  By definition of $\matS_p$,
  the loss of $D_P$ is at least $\matS_{n-12}[w_P,w'_P]$. Similarly, we
  have $\ell(D_Q)\geq \matS_{m-12}[w_Q,w'_Q]$, $\ell(D_R)\geq
  \matS_{n-12}[w_R,w'_R]$ and $\ell(D_O)\geq \matS_{m-12}[w_O,w'_O]$.
  By the definition of the loss:
  \begin{align} \ell(D) = {} & b_{n,m} \nonumber\\ = {} & 5\times
    \vert D \vert - \vert N[D]\vert \nonumber\\ = {} & \ell(D_O)
    +\ell(D_P) +\ell(D_Q) +\ell(D_R) + \matL[w'_O,w_P] + \matL[w'_P,w_Q] +
    \matL[w'_Q,w_R] + \matL[w'_R,w_O] \nonumber\\ & \text{~~ by
      Lemma~\ref{lem:compatible} and since $N[D_P] \cap N[D_R] = N[D_Q] \cap
      N[D_O] =\emptyset$} \nonumber\\ \ge {} & \matS_{m-12}[w_O,w'_O] +
    \matS_{n-12}[w_P,w'_P] + \matS_{m-12}[w_Q,w'_Q] +
    \matS_{n-12}[w_R,w'_R] \nonumber\\ & + \matL[w'_O,w_P] +
    \matL[w'_P,w_Q] + \matL[w'_Q,w_R] + \matL[w'_R,w_O] \nonumber\\ \ge {}
    & \matAS_{m-12}[w_O,w_P] + \matAS_{n-12}[w_P,w_Q] +
    \matAS_{m-12}[w_Q,w_R] + \matAS_{n-12}[w_R,w_O] \nonumber \\
    & \text{~~ since $w'_O$ and $w_P$ (resp. $w'_P$ and $w_Q$, $w'_Q$
      and $w_R$, $w'_R$ and $w_O$) are compatibles.} \nonumber
  \end{align}
\end{proof}

\medskip

According to Lemma~\ref{lem:bnm}, to bound $b_{n,m}$ it would be thus
interesting to know $\matAS_p$ for $p> 12$. It is why we introduce the
following $8119\times 8119$ square matrix, $\matT$.

\begin{lemma}
  There exists a matrix $\matT$ such that $\matS_{p+1}=\matS_p \otimes \matT$
  for all $p\ge 12$. This matrix is defined as follows:
  $$ \matT[w,w']=
  \begin{cases} 
    +\infty \ \ \ \text{ if $\exists i\in [10]$ s.t. $w[i]=0$ and $w'[i]=2$}\\ 
    +\infty \ \ \ \text{ if $\exists i\in [9]$ s.t. $w[i]=2$ and $w'[i]\ne 0$}\\ 
    +\infty \ \ \ \text{ if $\exists i\in \{2,\ldots, 9\}$
      s.t. $w'[i]=1$, $w[i]\ne 0$, $w'[i-1]\ne 0$ and $w'[i+1]\ne
      0$}\\  
    +\infty \ \ \ \text{ if $w'[1]=1$, $w[1]\ne 0$ and $w'[2]\ne 0$}\\ 
    +\infty \ \ \ \text{ if $w'[10]=1$, $w[10]\ne 0$ and $w'[9]\ne 0$}\\ 
    3 \times \vert w'\vert_0 - \vert w\vert_2 - \vert w' \vert_1 +
    \vert w\vert_0 - 1 \ \ \
    \text{if $w'[10] = 0$} \\
    3 \times \vert w'\vert_0 - \vert w\vert_2 - \vert w' \vert_1 +
    \vert w\vert_0 \ \ \ \text{ otherwise.}
  \end{cases}
  $$
\end{lemma}

\begin{proof} Consider a set $S'\subseteq P_{p+1}$ dominating
  $I(P_{p+1})$ and let $S =S' \cap P_p$.  Let $w = \wout_p(S)$ and $w' =
  \wout_{p+1}(S')$.  Let $\Delta(S,S')= \los{S'} - \los{S}$. By
  definition of the loss, $\Delta(S,S')= 5 \times |S'\setminus S| -
  |N[S']\setminus N[S]|$. Let us compute $\Delta(S,S')$ in term of the
  number of occurrences of $0$'s, $1$'s and $2$'s in the words $w$ and
  $w'$. The set $S'\setminus S$ corresponds to the vertices $\{(p+1,i)
  \mid i\in[10], w'[i] = 0\}$. The set $N[S']\setminus N[S]$ corresponds
  to the vertices dominated by $S'$ but not dominated by $S$; these
  vertices clearly belong to the columns $p$, $p+1$ and $p+2$. Since
  $S'$ dominates $I(P_{p+1})$, those in the column $p$ are the vertices
  $\{(p,i) \mid i\in[10], w[i]=2\}$.  Those in the column $p+1$ are the
  vertices $\{(p+1,i) \mid i\in[10], w'[i]\neq 2, w[i] \neq0\}$ and
  possibly the vertex $(p+1,11)$ when $w'[10]=0$. Finally, those in the
  column $p+2$ are the vertices $\{(p+2,i) \mid i\in[10], w'[i]=0\}$.
  We then get:
  $$ 
  \Delta(S,S')= 
  \begin{cases}
    3 \times \vert w'\vert_0 - \vert w \vert_2  
    - \vert w'\vert_1 + \vert w\vert_0 - 1 & \text{if $w'[10] = 0$} \\    
    3 \times \vert w'\vert_0 - \vert w \vert_2  
    - \vert w'\vert_1 + \vert w\vert_0 & \text{otherwise}
  \end{cases}
  $$
  where $|w|_n$ denotes the number of occurrences of the letter $n$ in
  the word $w$.

  Thus $\Delta(S,S')$ only depends on the output words of $S$ and $S'$,
  and we can denote this value by $\Delta(w,w')$. Note however that
  there exist pairs of
  words $(w,w')$ that could not be the output words of $S$ and $S'$;
  there are three cases:
  \begin{enumerate}[C{a}se 1.]
  \item $w[i]=0$ and $w'[i]=2$ since the vertex
    $(p+1,i)$ would be dominated by $(p,i)$ contradicting its label 2;
  \item $w[i]=2$ and $w'[i]\ne 0$ for $i\in[9]$ since 
    $(p,i)$ would not be dominated contradict the fact that $S'$
    dominates $I(P_{p+1})$;
  \item $w'[i]=1$ and $w'[i-1] \neq 0$, $w'[i+1]\neq 0$, $w[i]\neq 0$
    since $(p+1,i)$ would be dominated according to its label but none
    of its neighbors belong to $S'$. 
  \end{enumerate}
  For these forbidden cases, we set $\Delta(w,w') = +\infty$.
  
  By definition, $\matS_{p+1}[w_i,w']$ is the minimum loss $\ell(S')$
  of a set $S'\subseteq P_{p+1}$ that dominates $I(P_{p+1})$, with $w_i$
  as input word and $w'$ as output word. It is clear that $S= S'\cap P_p$
  dominates $I(P_{p})$ and has $w_i$ as input word.  Let $w$ be its output
  word and note that $\matS_{p+1}[w_i,w'] = \ell(S') = \ell(S) +
  \Delta(w_i,w')$.  The minimality of $\ell(S')$ implies the minimality of
  $\ell(S)$ over the sets $X\in \mD^{w_i,w'}_p$. Indeed, any set $X\in
  \mD^{w_i,w'}_p$ could be turned in a set of $X'\in \mD^{w_i,w'}_{p+1}$ by
  adding vertices of the $p+1^{th}$ column accordingly to $w'$.  Thus
  $$
  \matS_{p+1}[w_i,w'] = \matS_{p}[w_i,w] + \Delta(w,w')
  $$

  which implies that
  $$
  \matS_{p+1}[w_i,w'] \geq \min_{w} \matS_{p}[w_i,w] + \Delta(w,w').
  $$

  On the other hand, for every word $w_o \in \mW$ such that
  $\matS_{p}[w_i,w_o]\neq +\infty$ and $\Delta(w_o,w')\neq
  +\infty$, there is a set $S \in \mD^{w_i,w_o}_p$, with $\ell(S) =
  \matS_{p}[w_i,w_o]$, that can be turned in a set $S' \in
  D^{w_i,w'}_{p+1}$ with $\ell(S') = \matS_{p}[w_i,w_o] +
  \Delta(w_o,w')$. Thus
  $$
  \matS_{p+1}[w_i,w'] \leq \min_{w_o} \matS_{p}[w_i,w_o] + \Delta(w_o,w').
  $$

  This concludes the proof of the lemma.
\end{proof}


By the definition of $M_p$, we have also $\matAS_{p+1}=\matAS_p
\otimes \matT$. Note that $\matT$ is a sparse matrix: about $95.5 \%$
of its $8119^2$ entries are $+\infty$.  Thus the multiplication by
$\matT$ in the $(\min,+)$ algebra can be done in a reasonable amount
of time by a trivial algorithm.

\medskip

\begin{fact}\label{fact126} The computations give us that
  $\matAS_{126}=\matAS_{125}+1$. Thus, since $(A + c)\otimes B =
  (A\otimes B) + c$ for any matrices $A$, $B$ and any integer $c$, we
  have that $\matAS_{125+k}=\matAS_{125}+k$ for every $k\in \mathbb{N}$.
\end{fact}

Let us define $\matMin_p= \min_{k\in \mathbb{N}} ( \matAS_{p+k} - k
)$.  Then, for all $q\ge p$, $\matAS_q \ge \matMin_p + (q-p)$.  By
Fact~\ref{fact126}, $\matMin_p= \min_{k\in \{0,\dots 125-p\}} (
\matAS_{p+k} - k )$

\begin{fact}\label{fact23} By computing $\matMin_{12}$, and $A'=
  \matMin_{12} \otimes \matMin_{12}$, we obtain that $\min_{w_1,w_3} (A'
  + A'^T)[w_1,w_3] = 76$ (where $A^T$ is the transpose of $A$).
\end{fact}
This implies that 
$$ 
\min_{w_1,w_3} \ (\min_{w_2} \matMin_{12}[w_1,w_2]+\matMin_{12}[w_2,w_3]) \ +
\ (\min_{w_4} \matMin_{12}[w_3,w_4]+ \matMin_{12}[w_4,w_1]) \ = \ 76
$$
$$
\min_{w_1,w_2,w_3,w_4} \matMin_{12}[w_1,w_2]+\matMin_{12}[w_2,w_3]+\matMin_{12}[w_3,w_4]+
\matMin_{12}[w_4,w_1] \ = \ 76.
$$

\begin{theorem}
  If $24 \le n \le m$, then $$\gamma(\Gnm) = \form{n}{m}.$$
\end{theorem}

\begin{proof} By Chang's construction \cite{CHHW}, $\gamma(\Gnm) \le
  \form{n}{m}$. Let us now compute a lower bound for the loss of a
  dominating set of $\Gnm$.

  \begin{eqnarray*}
    \lnm &\ge& \bnm\\
    &\ge& \min_{w_1,w_2,w_3,w_4}  \matAS_{n-12}[w_1,w_2] +
    \matAS_{m-12}[w_2,w_3] + \matAS_{n-12}[w_3,w_4]+
    \matAS_{m-12}[w_4,w_1]\\
    & & \; \mbox{by Lemma~\ref{lem:bnm}} \\
    &\ge& \min_{w_1,w_2,w_3,w_4}  \matMin_{12}[w_1,w_2]+(n-12-12)  +
    \matMin_{12}[w_2,w_3]+ (m-12-12) + \matMin_{12}[w_3,w_4]\\&&
    \phantom{\min_{w_1,w_2,w_3,w_4}}+(n-12-12) +
    \matMin_{12}[w_4,w_1]+ (m-12-12)\\ 
    &\ge& 2 \times (n+m -48) + \min_{w_1,w_2,w_3,w_4}
    \matMin_{12}[w_1,w_2] + \matMin_{12}[w_2,w_3] +
    \matMin_{12}[w_3,w_4] + \matMin_{12}[w_4,w_1]\\ 
    &\ge& 2 \times (n+m -48) + 76\\ 
    &\ge& 2\times(n+m)-20.
  \end{eqnarray*}

  Thus by Lemma~\ref{lemloss}, we have:
  \begin{eqnarray*}
    \gamma(\Gnm) &\ge& \left \lceil \frac{n\times m + 2\times(n+m)-20}{5}
    \right \rceil\\
    &\ge& \left \lceil \frac{(n+2)(m + 2) -4}{5}
    \right \rceil - 4\\
    &\ge& \form{n}{m}.
  \end{eqnarray*}
\end{proof}


\begin{thebibliography}{99}

\bibitem{Chang}
  {\sc T.Y. Chang}, 
  {Domination Numbers of Grid Graphs},
  {\sl Ph.D. Thesis}, Department of Mathematics, University of South Florida, 1992.

\bibitem{CC}
  {\sc T.Y. Chang, W.E. Clark},
  {The domination numbers of the $5\times n$ and $6\times n$ grid graphs},
  {\sl J. Graph Theory}, {\bf 17} (1993) 81-108.

\bibitem{CCH}
  {\sc T.Y. Chang, W.E. Clark, E.O. Hare},
  {Dominations of complete grid graphs I},
  {\sl Ars Combin.}, {\bf 38} (1994) 97-112.

\bibitem{CHHW}
  {\sc E.J. Cockayne, E.O. Hare, S.T. Hedetniemi, T.V. Wimer}, 
  {Bounds for the Domination Number of Grid Graphs},
  {\sl Congressus Numeratium} {\bf 47} (1985) 217-228.

\bibitem{Fisher}
  {\sc David C. Fisher},
  The Domination number of complete grid graphs,
  {\sl manuscript}

\bibitem{Guichard}
  {\sc David R. Guichard}, 
  A lower bound for the domination number of complete grid graphs,
  {\sl J. Combin. Math. Combin. Comput.}  49  (2004), 215--220.

\bibitem{JK}
  {\sc M. S. Jacobson, L. F. Kinch},
  On the domination number of products of a graph; I,
  {\sl Ars Comb.} {\bf 10} (1983), pp. 33-44.

\bibitem{HHH}
  {\sc E.O. Hare, S.T. Hedetniemi, W.R. Hare}, 
  {Algorithms for Computing the Domination Number of $K\times N$ Complete Grid Graphs},
  {\sl Congressus Numeratium} {\bf 55} (1986) 81-92.

\bibitem{SP}
  {\sc H.G. Singh, and R.P. Pargas},
  {A parallel implementation for the domination number of grid graphs},
  {\sl Cong. Numer.} {\bf 59} (1987) 297-312.

\end{thebibliography}
\end{document}